\documentclass[a4paper]{article}

\usepackage{fullpage}

\usepackage[utf8]{inputenc}
\usepackage[T1]{fontenc}

\usepackage{todonotes}
\usepackage{xspace}
\usepackage{tikz}
\usetikzlibrary{calc}
\usepackage[sort&compress,numbers]{natbib}

\newcommand{\Oh}{\mathcal{O}}

\newcommand{\N}{\mathbb{N}}

\newcommand{\NP}{\textsf{NP}\xspace}
\newcommand{\PP}{\textsf{P}\xspace}

\newcommand{\MWIS}{\textsc{MWIS}\xspace}

\renewcommand{\leq}{\leqslant}
\renewcommand{\geq}{\geqslant}
\renewcommand{\epsilon}{\varepsilon}

\usepackage{framed}
\usepackage{tabularx}

\newlength{\RoundedBoxWidth}
\newsavebox{\GrayRoundedBox}
\newenvironment{GrayBox}[1]%
   {\setlength{\RoundedBoxWidth}{.93\columnwidth}
    \def\boxheading{#1}
    \begin{lrbox}{\GrayRoundedBox}
       \begin{minipage}{\RoundedBoxWidth}}%
   {   \end{minipage}
    \end{lrbox}
    \begin{center}
    \begin{tikzpicture}%
       \node(Text)[draw=black!20,fill=white,rounded corners,inner sep=2ex,text width=\RoundedBoxWidth]
             {\usebox{\GrayRoundedBox}};
        \coordinate(x) at (current bounding box.north west);
        \node [draw=white,rectangle,inner sep=3pt,anchor=north west,fill=white]
        at ($(x)+(6pt,.75em)$) {\boxheading};
    \end{tikzpicture}
    \end{center}}

\newenvironment{defproblemx}[1]{\noindent\ignorespaces%
                                \FrameSep=6pt%
                                \parindent=0pt%
                \begin{GrayBox}{#1}%
                \begin{tabular*}{\columnwidth}{!{\extracolsep{\fill}}@{\hspace{.1em}} >{\itshape} p{1.5cm} p{0.86\columnwidth} @{}}%
            }{
                \end{tabular*}%
                \end{GrayBox}%
                \ignorespacesafterend
            }

\newcommand{\problemTask}[3]{%
  \begin{defproblemx}{#1}
    Input: & #2 \\
    Task: & #3
  \end{defproblemx}
}

\usepackage{amsmath,amssymb}
\usepackage{thmtools}
\usepackage{thm-restate}
\usepackage[ocgcolorlinks, linkcolor={blue}, citecolor={brown}]{hyperref}
\usepackage[amsmath,thmmarks,hyperref]{ntheorem}
\usepackage{cleveref}

\newtheorem{theorem}{Theorem}
\crefformat{theorem}{#2Theorem~#1#3}
\Crefformat{theorem}{#2Theorem~#1#3}

\newcommand{\newtheoremwithcrefformat}[2]{%
  \newtheorem{#1}[theorem]{#2}%
  \crefformat{#1}{##2\MakeUppercase#1~##1##3}%
  \Crefformat{#1}{##2\MakeUppercase#1~##1##3}%
}
\newcommand{\newseptheoremwithcrefformat}[2]{%
  \newtheorem{#1}{#2}%
  \crefformat{#1}{##2\MakeUppercase#1~##1##3}%
  \Crefformat{#1}{##2\MakeUppercase#1~##1##3}%
}

\newtheoremwithcrefformat{lemma}{Lemma}
\newtheoremwithcrefformat{proposition}{Proposition}
\newtheoremwithcrefformat{observation}{Observation}
\newtheoremwithcrefformat{conjecture}{Conjecture}
\newtheoremwithcrefformat{corollary}{Corollary}
\newseptheoremwithcrefformat{claim}{Claim}
\theorembodyfont{\upshape}
\newtheoremwithcrefformat{example}{Example}
\newtheoremwithcrefformat{remark}{Remark}
\theoremstyle{nonumberplain}
\theoremheaderfont{\scshape}
\theorembodyfont{\normalfont}
\theoremsymbol{\ensuremath{\square}}
\newtheorem{proof}{Proof.}

\theoremsymbol{\ensuremath{\lrcorner}}
\newtheorem{claimproof}{Proof of Claim.}

\def\cqedsymbol{\ifmmode$\lrcorner$\else{\unskip\nobreak\hfil
\penalty50\hskip1em\null\nobreak\hfil$\lrcorner$
\parfillskip=0pt\finalhyphendemerits=0\endgraf}\fi}

\title{
Max Weight Independent Set in sparse graphs with no long claws\footnote{Parts of this work appeared in the proceedings of SODA 2021~\cite{ACDR21} and STACS 2024~\cite{stacspaper}}}
\author{
   Tara Abrishami\thanks{Department of Mathematics, University of Hamburg, Germany, \texttt{tara.abrishami@uni-hamburg.de}. Supported by the National Science Foundation Award Number DMS-2303251 and the Alexander von Humboldt Foundation. This work was performed in part while the author was at Princeton University and supported by NSF Grant DMS-1763817 and NSF-EPSRC Grant DMS-2120644.} \and
  Maria Chudnovsky\thanks{Department of Mathematics, Princeton University, USA, \texttt{mchudnov@math.princeton.edu}. Supported by NSF-EPSRC Grant DMS-2120644 and by AFOSR grant FA9550-22-1-008.} \and
  Cemil Dibek\thanks{Department of Operations Research and Financial Engineering, Princeton University, USA, \texttt{cdibek@alumni.princeton.edu}. Supported by the National Science Foundation Award Number DMS-1763817.} \and
  Marcin Pilipczuk\thanks{Institute of Informatics, University of Warsaw, Poland, \texttt{m.pilipczuk@mimuw.edu.pl}. Supported by Polish National Science Centre SONATA BIS-12 grant number 2022/46/E/ST6/00143.} \and
  Pawe\l{} Rz\k{a}\.{z}ewski\thanks{Faculty of Mathematics and Information Science, Warsaw University of Technology and Institute of Informatics, University of Warsaw, Poland, \texttt{pawel.rzazewski@pw.edu.pl}. 
    This work is a part of the project BOBR that has received funding from the European Research Council (ERC) under the European Union's Horizon 2020 research and innovation programme (grant agreement No.~948057)}}

\date{}

\usepackage{mathtools}

\newcommand{\weight}{\mathbf{w}}
\newcommand{\maxdeg}{\Delta}

\begin{document}

\maketitle

\begin{abstract}
For graphs $G$ and $H$, we say that $G$ is $H$-free if it does not contain $H$ as an induced subgraph.
Already in the early 1980s Alekseev observed that the \textsc{Max Weight Independent Set} problem (\textsc{MWIS}) remains \textsf{NP}-hard in $H$-free graphs, unless every component of $H$ is a path or a subdivided claw, i.e., a graph obtained from the three-leaf star by subdividing each edge some number of times (possibly zero). Since then determining the complexity of \textsc{MWIS} in these remaining cases is one of the most important problems in algorithmic graph theory.

In this paper we make an important step towards solving the problem by providing a polynomial-time algorithm for \textsc{MWIS} in graphs excluding a fixed graph forest of paths and subdivided claws as an induced subgraph, and a fixed biclique as a subgraph.

\end{abstract}

%\newpage

\section{Introduction}\label{sec:intro}
A \emph{vertex-weighted graph} is an undirected graph $G$ equipped with a weight function $\weight : V(G) \to \N$.
For $X \subseteq V(G)$, we use $\weight(X)$ as a shorthand for $\sum_{x \in X} \weight(x)$ and for a subgraph $H$ of $G$,
$\weight(H)$ is a shorthand for $\weight(V(H))$.  By convention we use $\weight(\emptyset)=0$.
Throughout the paper we assume that arithmetic operations on weights are performed in unit time.

For a graph $G$, a set $I \subseteq V(G)$ is \emph{independent} or \emph{stable} if there is no edge of $G$ with both endpoints in $I$. By $\alpha(G)$ we denote the number of vertices in a largest independent set in $G$.
In the \textsc{Max Independent Set} (\textsc{MIS}) problem we are given an undirected graph $G$, and ask for an independent set of size $\alpha(G)$.
In the \textsc{Max Weight Independent Set} (\textsc{MWIS}) problem we are given an undirected vertex-weighted graph $(G,\weight)$, and ask for a maximum-weight independent set in $(G,\weight)$. Note that \textsc{MIS} is a special case of \textsc{MWIS} where all weightes are equal. By $n$ we always denote the number of vertices of the instance graph.

\textsc{MIS} (and \textsc{MWIS} as its generalization) is a ``canonical'' hard problem: It was one of the first problems shown to be \NP-hard~\cite{Karp1972}, it is notoriously hard to approximate~\cite{Hastad96cliqueis,Khot2006}, and it is \textsc{W}[1]-hard~\cite{Cyganetal}. Many of these hardness results hold even if we restrict input instances to some natural graph classes~\cite{GAREY1976237,DBLP:conf/isaac/BonnetBTW19,DBLP:conf/wg/DvorakFRR20}.
Thus a very natural research direction is to consider restricted instances and try to capture the boundary between ``easy'' and ``hard'' cases.

\paragraph*{State of the art.}
The study of the complexity of \MWIS in restricted graph classes is a central topic in algorithmic graph theory~\cite{butterfly,GROTSCHEL1984325,DBLP:conf/soda/Yolov18,DBLP:journals/corr/abs-1912-11246,DBLP:journals/dam/Alekseev03,DBLP:conf/wg/BergougnouxKR23}.
Particular attention is given to classes that are \emph{hereditary}, i.e., closed under vertex deletion.
Among such classes a special role is played by ones defined by forbidding certain substructures. For graphs $G$ and $H$, we say that $G$ is \emph{$H$-free} if it does not contain $H$ as an \emph{induced subgraph}.

In what follows, for $t \geq 1$, by $P_t$ we denote the $t$-vertex path.
For integers $a,b,c \geq 1$, by $S_{a,b,c}$ we denote the graph obtained from the three-leaf star (i.e., the \emph{claw}) by subdividing the three edges $a-1$, $b-1$, and $c-1$ times, respectively.
Alternatively, we may think of $S_{a,b,c}$ as the graph obtained from paths $P_{a+1}, P_{b+1}$, and $P_{c+1}$ by identifying one endvertex of each path.
By $d S_{a,b,c}$ we denote the graph with $d$ components, each of which is isomorphic to $S_{a,b,c}$.

Let $\mathcal{S}$ be the family of subcubic forests $H$ whose every component has at most one vertex of degree 3,
i.e., whose every component is either a path or a subdivided claw.

The complexity study of \textsc{MWIS} in $H$-free graphs dates back to the early 1980s and the work of Alekseev~\cite{alekseev1982effect},
who observed that for most graphs $H$ the problem remains \NP-hard.
Indeed, let us discuss the hard cases.
First, \textsc{MIS} (and thus \textsc{MWIS}) is \NP-hard in subcubic graphs~\cite{GAREY1976237},
which are $H$-free whenever $H$ has a vertex of degree at least 4.
For the remaining cases we will use the so-called \emph{Poljak construction}~\cite{Po74}:
If $G'$ is obtained from $G$ by subdividing one edge twice, then $\alpha(G') = \alpha(G)+1$.
Thus, if $G^p$ denotes the graph obtained from $G$ by subdividing each edge exactly $2p$ times,
then $\alpha(G^p) = \alpha(G) + p \cdot |E(G)|$.
Now observe that if $H$ has a cycle or two vertices of degree three in one component,
then $G^{|V(H)|}$ is $H$-free. Consequently, for such graph $H$, \textsc{MIS} is \NP-hard in $H$-free graphs.
Let us point out that the above hardness reductions imply that \textsc{MIS} cannot even be solved in subexponential time unless the Exponential-Time Hypothesis (ETH) fails.

Summing up, the only graphs $H$ for which we may hope for a polynomial-time (or even subexponential-time) algorithm for \textsc{MWIS} in $H$-free graphs are precisely the graphs in $\mathcal{S}$.

The complexity of \textsc{MWIS} in $H$-free graphs when $H \in \mathcal{S}$ remains one of the most challenging and important problems in algorithmic graph theory. Despite significant attention received from the graph theory and the theoretical computer science communities, only partial results are known. Let us discuss them.

First, consider the case that $H = P_t$ for some $t$.
Since $P_4$-free graphs, also known as \emph{cographs}, have very rigid structure (in particular, they have clique-width at most 2), the polynomial-time algorithm for this class of graphs is rather simple~\cite{CORNEIL1981163}.
However, already for $P_5$-free graphs the situation is much more complicated.
The existence of a polynomial-time algorithm for $P_5$-free graphs was a long-standing open problem that was finally resolved in the affirmative in 2014 by Lokshtanov, Vatshelle, and Villanger~\cite{DBLP:conf/soda/LokshantovVV14} using the framework of \emph{potential maximal cliques}.
Later, using the same approach but with significantly more technical effort, Grzesik, Klimo\v{s}ov\'a, Pilipczuk, and Pilipczuk~\cite{DBLP:journals/talg/GrzesikKPP22} obtained a polynomial-time algorithm for $P_6$-free graphs.
The case of $P_7$-free graphs remains open. 

However, some interesting algorithmic results can be obtained if we relax our notion of an efficient algorithm.
First, it was shown by Bacs\'o et al.~\cite{DBLP:journals/algorithmica/BacsoLMPTL19} that for every fixed $t$, \textsc{MWIS} can be solved in \emph{subexponential} time $2^{\Oh(\sqrt{n \log n})}$ for $P_t$-free graphs (by $n$ we always denote the number of vertices of the instance graph).
Another subexponential-time algorithm, with worse running time, was obtained independently by Brause~\cite{DBLP:journals/dam/Brause17}.
While these results do not rule out the possibility that the problem is \NP-hard,
let us recall that, assuming the ETH, subexponential algorithms for \textsc{MWIS} in $H$-free graphs cannot exist if $H \notin \mathcal{S}$.
Later, Chudnovsky et al.~\cite{DBLP:conf/soda/ChudnovskyPPT20,DBLP:journals/corr/abs-1907-04585} showed that for every fixed $t$, the problem admits a QPTAS in $P_t$-free graphs.
Finally, a very recent breakthrough result by Gartland and Lokshtanov~\cite{GartlandL20} shows that for every fixed $t$, the problem can be solved in \emph{quasipolynomial time} $n^{\Oh(\log^3 n)}$; see also a slightly simpler algorithm by Pilipczuk, Pilipczuk, and Rz\k{a}\.zewski~\cite{DBLP:conf/sosa/PilipczukPR21} with running time $n^{\Oh(\log^2n)}$.
Note that this means that if for some $t$,\textsc{MWIS} is \NP-hard for $P_t$-free graphs, then \emph{all problems} in \NP can be solved in quasipolynomial time. While this does yet not imply that \PP = \NP, it still seems rather unlikely according to our current understanding of complexity theory.

Now let us turn to the case when $H$ is a subdivided claw.
The simplest subdivided claw is the claw $S_{1,1,1}=K_{1,3}$. Claw-free graphs appear to be closely related to \emph{line graphs}~\cite{Claw5} and thus a polynomial-time algorithm for \textsc{MWIS} in claw-free graphs can be obtained by a modification of the well-known augmenting path approach for finding a maximum-weight matching~\cite{SBIHI198053,MINTY1980284} (i.e., a maximum-weight independent set in a line graph).
Let us highlight the close relation of claw-free graphs and line graphs, as it will play an important role in our paper.
The next smallest subdivided claw is the \emph{fork}, i.e., $S_{2,1,1}$. A polynomial-time algorithm for \textsc{MIS} in fork-free graphs was obtained by Alekseev~\cite{ALEKSEEV20043}. Later it was extended to the \textsc{MWIS} problem by Lozin and Milani\v{c}~\cite{DBLP:journals/jda/LozinM08}. For disconnected $H$, it is known that \textsc{MWIS} is polynomial-time-solvable in $dS_{1,1,1}$-free graphs, for every fixed $d$~\cite{BrandstadtM18a}.
The existence of polynomial-time algorithms in the next simplest (connected) cases, i.e., $H = S_{3,1,1}$ and $H=S_{2,2,1}$, is wide open.

Again, some interesting results can be obtained if we look beyond polynomial-time algorithms.
Chudnovsky et al.~\cite{DBLP:conf/soda/ChudnovskyPPT20,DBLP:journals/corr/abs-1907-04585} proved that for every subdivided claw $H$, the \textsc{MWIS} problem in $H$-free graphs admits a QPTAS and a subexponential-time algorithm working in time $n^{\Oh(n^{8/9})}$. We point out that the arguments used for the case when $H$ is a subdivided claw are significantly more complicated and technically involved than their counterparts for $P_t$-free graphs.
These results were then simplified and improved by Majewski et al.~\cite{ICALP-qptas}: They obtained another (faster) QPTAS and a subexponential-time algorithm with running time $n^{\Oh(\sqrt{n} \log n)}$.
Finally, very recently, Gartland et al.~ \cite{arxiv-sttt-qpoly} announced a quasipolynomial-time algorithm for \textsc{MWIS} in $H$-free graphs for every $H \in \mathcal{S}$.

More tractability results can be obtained if we put some additional restrictions on the instance graph. In particular,
there is a long line of research concerning graphs excluding a fixed (but still small) path or a subdivided claw and, simultaneously, some other small graphs, see e.g.~\cite{DBLP:journals/endm/LeBS15,p71,p67,p65,p73,p72,DBLP:journals/gc/LozinMP14,DBLP:journals/jda/LozinMR15,DBLP:journals/tcs/HarutyunyanLLM20,Mosca99,Mosca09,Mosca13,Mosca21,DBLP:journals/tcs/BrandstadtM21}.
A slightly different direction was considered by Lozin, Milani\v{c}, and Purcell~\cite{DBLP:journals/gc/LozinMP14}, who proved that for every fixed $t$,\textsc{MWIS} is polynomial-time solvable in \emph{subcubic} $S_{t,t,1}$-free graphs. Later, Lozin, Monnot, and Ries~\cite{DBLP:journals/jda/LozinMR15} showed a polynomial time algorithm for subcubic $S_{2,2,2}$-free graphs. Finally, Harutyunya et al.~\cite{DBLP:journals/tcs/HarutyunyanLLM20} generalized both these results by providing a polynomial-time algorithm for subcubic $S_{2,t,t}$-free graphs, for any fixed $t$.

We remark that the case when $H$ is a subdivided claw (or, more precisely, is in $\mathcal{S}$ and contains at least one component which is not a path) is the only case where the restriction to bounded degree graphs leads to an interesting problem. Indeed, the already mentioned hardness reduction of Alekseev~\cite{alekseev1982effect} shows that if $H \notin \mathcal{S}$, then \textsc{MIS} is \NP-hard even in \emph{subcubic} $H$-free graphs. On the other hand, if $H$ is a forest of paths, then $H$-free graphs of bounded degree are of constant size and thus of little interest.

In this work, we continue and significantly extend the study of the complexity of \textsc{MWIS} in $H$-free graphs with additional restrictions, where $H \in \mathcal{S}$.

\paragraph*{Our results.}
As a warm-up, we present  a polynomial-time algorithm for $H$-free graphs of bounded degree, where $H \in \mathcal{S}$.

\begin{theorem}\label{thm:main-maxdeg}
There exists an algorithm that, 
given a vertex-weighted graph $(G,\weight)$ on $n$ vertices with maximum degree $\maxdeg$ and integers $d,t$
in time $2^{\Oh(dt \Delta ^2)}n^{\Oh(t\Delta^2)}$ either finds an induced $dS_{t,t,t}$ or the maximum possible weight of an independent set in $(G,\weight)$.
\end{theorem}
Note that by picking appropriate $d$ and $t$, \cref{thm:main-maxdeg} yields a polynomial-time algorithm for \textsc{MWIS} for bounded-degree graphs excluding a fixed graph from $\mathcal{S}$ as an induced subgraph.

Then we proceed to the main result of the paper: we show that \MWIS remains polynomial-time solvable in $dS_{t,t,t}$-free graphs, even if instead of bounding the maximum degree, we forbid a fixed biclique \emph{as a subgraph}.

\begin{theorem}\label{thm:main-biclique}
For every fixed integers $d,t$, and $s$ there exists a polynomial-time algorithm that, given 
a vertex-weighted graph $(G,\weight)$ that does not contain $dS_{t,t,t}$ as an induced subgraph nor $K_{s,s}$ as a subgraph, returns the maximum possible weight of an independent set in $(G,\weight)$. 
\end{theorem}

Let us remark that by the celebrated K\"{o}v\'ari-S\'os-Tur\'an theorem~\cite{Kovari1954}, classes that exclude $K_{s,s}$  as a subgraph capture all hereditary classes of \emph{sparse graphs}, where by ``sparse'' we mean that the graph has a subquadratic number of edges. 
Furthermore, by a simple Ramsey argument, for every positive integer $r$ there exists an integer $s$ such that
if $G$ is $K_r$-free and $K_{r,r}$-free then $G$ does not contain $K_{s,s}$ as a subgraph. 
Hence, equivalently, \cref{thm:main-biclique} yields a polynomial-time algorithm for \textsc{MWIS} for graphs that are simultaneously $H$-free (for some $H \in \mathcal{S}$), $K_r$-free, and $K_{r,r}$-free.

\paragraph*{Our techniques.}
As in the previous works~\cite{DBLP:conf/soda/ChudnovskyPPT20,arxiv-sttt-qpoly}, the crucial tool in handling $dS_{t,t,t}$-free graphs
is an \emph{extended strip decomposition}. 
Its technical definition can be found in preliminaries; for now, it suffices to say that
it is a wide generalization of the preimage graph of a line graph 
(recall that line graphs are $S_{1,1,1}$-free) that allows for recursion for the 
\textsc{MWIS} problem. An extended strip decomposition of a graph $G$ identifies some induced
subgraphs of $G$ as \emph{particles} and, knowing the maximum possible weight of an independent
set in each particle, one can compute
in polynomial time
the maximum possible weight of an independent set in $G$.
(We remark that this computation involves advanced combinatorial techniques as it relies
on a reduction to the maximum weight matching problem in an auxiliary graph.)
In other words, finding an extended strip decomposition with small particles 
compared to $|V(G)|$ is equally good for the \textsc{MWIS} problem as splitting the graph
into small connected components.

The starting point is the following theorem of~\cite{ICALP-qptas}.
\begin{theorem}[{\cite[Corollary 12]{ICALP-qptas}} in a semi-weighted setting]\label{thm:ICALP:weight}
There exists an algorithm that, given an $n$-vertex graph $G$ with a set $U \subseteq V(G)$
and integers $d,t$, in polynomial time outputs
either:
\begin{itemize}
\item an induced copy of $dS_{t,t,t}$ in $G$, or
\item a set $X$ of size at most $(d-1)(3t+1)+(11 \log n + 6)(t+1)$ and a rigid extended strip decomposition of
  $G - N[X]$ with every particle containing at most $|U|/2$ vertices of $U$.
\end{itemize}
\end{theorem}
(A rigid extended strip decomposition is an extended strip decomposition that does not  have some unnecessary empty sets. By $N[X]$ we denote the set consisting of $X$ and all vertices with a neighbor in $X$.)
Let us remark that the result stated in \cite[Theorem~2]{ICALP-qptas} is for unweighted graphs (i.e., $U=V(G)$ using the notation from \cref{thm:ICALP:weight}), but the statement of \cref{thm:ICALP:weight} can be easily derived from the proof, see also~\cite{arxiv-sttt-qpoly}.

Consider the setting of \cref{thm:main-maxdeg}, i.e., the graph $G$ has maximum degree
$\maxdeg$. Apply \cref{thm:ICALP:weight} to $G$ with $U=V(G)$.
If we get the first outcome, i.e., an induced $dS_{t,t,t}$ in $G$, we return it and terminate.
So assume that we get the second outcome, i.e., the set $X$.
Note that as $|X| = \Oh(dt + t\log n)$, we have $|N[X]| = \Oh(dt\maxdeg + t\maxdeg\log n)$. 
It is now tempting to exhaustively branch on $N[X]$ (i.e., guess the intersection of the sought
independent set with $N[X]$) and recurse on the particles of the extended strip decomposition
of $G-N[X]$. However, implementing this strategy directly gives quasipolynomial (in $n$) running
time bound of $n^{\Oh(dt\maxdeg + t \maxdeg \log n)}$, as the branching step 
yields up to $2^{|N[X]|} = 2^{\Oh(dt \Delta)} \cdot n^{\Oh(t \maxdeg)}$ subcases and the depth of the recursion
is $\Oh(\log n)$. 

Our main new idea now is to perform this branching lazily, by considering a more general 
\emph{border} version of the problem, where the input graph is additionally equipped
with a set of \emph{terminals} and we ask for a maximum weight of an independent set
for every possible behavior on the terminals.
\problemTask{\textsc{Border MWIS}}{A vertex-weighted graph $(G,\weight)$ with a set $T \subseteq V(G)$ of \emph{terminals}.}{Compute $f_{G,\weight,T} : 2^T \to \mathbb{N} \cup \{-\infty\}$ defined for every
  $I_T \subseteq T$ as \newline
 $ f_{G,\weight,T}(I_T) = \max \{ \weight(I)~|~I \subseteq V(G) \wedge I\mathrm{\ is\ independent} \wedge I \cap T = I_T\}.$
}

A similar application of a border version of the problem to postpone branching in recursion
appeared for example in the technique of recursive understanding~\cite{DBLP:conf/focs/KawarabayashiT11,DBLP:journals/siamcomp/ChitnisCHPP16}.

Let us return to our setting, where we have a set $X$ of size $\Oh(dt + t \log n)$
and an extended strip decomposition of $G-N[X]$ with particles of size at most half of the
size of $V(G)$. 
We would like to remove $N[X]$ from the graph, indicate $N(N[X])$ as terminals and solve
\textsc{Border MWIS} in $(G-N[X],\weight,T := N(N[X]))$ using the extended strip decomposition
for recursion. Note that, thanks to the bounded degree assumption,
the size of $T = N(N[X])$ is bounded by $\Oh(dt \maxdeg^2 + t \maxdeg^2 \log n)$.
    
This approach \emph{almost} works: the only problem is that, as the recursion progresses, 
the set of terminals accummulates and its size can grow beyond the initial $\Oh(dt \maxdeg^2 + t \maxdeg^2 \log n)$
bound. Luckily, this can be remedied in a standard way: we alternate recursive steps where
\cref{thm:ICALP:weight} is invoked with $U=V(G)$ with steps
where \cref{thm:ICALP:weight} is invoked with $U=T$.
In this manner, we can maintain a bound of $\Oh(dt \maxdeg^2 + t \maxdeg^2 \log n)$ on the number of terminals
in every recursive call.
Note that this bound also guarantees that the size of the domain of the
requested function $f_{G,\weight,T}$ is of size $2^{\Oh(dt \maxdeg^2)} n^{\Oh(t \maxdeg^2)}$, which is within
the promised time bound.

\medskip

Let us now move to the more general setting of \cref{thm:main-biclique}.
Here, the starting points are the recent results of
Wei{\ss}auer~\cite{DBLP:journals/siamdm/Weissauer19} and
Lozin and Razgon~\cite{DBLP:journals/ejc/LozinR22}
that show that in the $S_{t,t,t}$-free case, 
excluding a biclique as a subgraph is not that much different than bounding the maximum degree.

A \emph{$k$-block} in a graph is a set of $k$ vertices, no two of which can be separated by deleting fewer than $k$ vertices.
The following result was shown by Wei{\ss}auer (we refer to preliminaries for standard
    definitions of tree decompositions and torsos).

\begin{theorem}[Wei{\ss}auer~\cite{DBLP:journals/siamdm/Weissauer19}]\label{thm:weissauer}
Let $G$ be a graph and $k\geq 2$ be an integer.
If $G$ has no $(k+1)$-block, then $G$ admits a tree decomposition with adhesion less than $k$, in which every torso has at most $k$ vertices of degree larger than $2k(k-1)$.
\end{theorem}
Even though the statement of the result in~\cite{DBLP:journals/siamdm/Weissauer19} is just existential, the proof actually yields a polynomial-time algorithm to compute such a tree decomposition.

It turns out that $dS_{t,t,t}$-free graphs with no large bicliques have no large blocks.
\begin{lemma}\label{lem:lozinrazgon}
For any $d,t$, and $s$ there exists $k$ such that the following holds.
Every $dS_{t,t,t}$-free graph with no subgraph isomorphic to $K_{s,s}$ has no $k$-block.
\end{lemma}
Let us remark that Lozin and Razgon~\cite{DBLP:journals/ejc/LozinR22} showed \cref{lem:lozinrazgon} 
for $S_{t,t,t}$-free graphs. However, an extension of their argument applies to $dS_{t,t,t}$-free graphs; we include it in  \Cref{app:lozinrazgon}.

Combining \cref{thm:weissauer} and \cref{lem:lozinrazgon} we immediately obtain the following.
\begin{corollary}\label{cor:tree-decomp}
For any $d,t$, and $s$ there exists $k$ such that the following holds.
Given a $dS_{t,t,t}$-free graph $G$ with no subgraph isomorphic to $K_{s,s}$,
in polynomial time one can compute a tree decomposition of $G$ with adhesion less than $k$, in which every torso has at most $k$ vertices of degree larger than $2k(k-1)$.
\end{corollary}
To prove \cref{thm:main-biclique} using \cref{cor:tree-decomp}
we need to carefully combine explicit branching on the (bounded number of) vertices of 
large degree in a single bag with --- as in the bounded degree case --- applying
\cref{thm:ICALP:weight} to the remainder of the graph and indicating
$N(N[X])$ as the terminal set of the border problem passed to the recursive
calls. 
Finally, we combine these steps with the information passed over
adhesions in the tree decomposition.

\section{Preliminaries}\label{sec:prelim}
\newcommand{\tree}{\mathcal{T}}
Our algorithms take a vertex-weighted graph $(G,\weight)$ as an input.
In the recursion, we will be working on various induced subgraphs of $G$
with vertex weight inherited from $\weight$.
Somewhat abusing notation, we will keep $\weight$ for the weight function
in any induced subgraph of $G$.

\subparagraph{Tree decompositions.}
Let $G$ be a graph. A \emph{tree decomposition} of $G$ is a pair $(\tree,\beta)$ where
$\tree$ is a tree and $\beta : V(\tree) \to 2^{V(G)}$ is a function satisfying
the following: (i) for every $uv \in E(G)$ there exists $t \in V(\tree)$ with $u,v \in \beta(t)$, and
(ii) for every $v \in V(G)$ the set $\{t \in V(\tree)~|~v \in \beta(t)\}$ induces a connected nonempty subtree of $\tree$. 
For every $t \in V(\tree)$ and $st \in E(\tree)$, the set $\beta(t)$ is the \emph{bag} at node $t$ and the set $\sigma(st) := \beta(s) \cap \beta(t)$ is the \emph{adhesion} at edge $st$. 
The critical property of a tree decomposition $(\tree,\beta)$ is that if $st \in E(\tree)$ and $V_s$ and $V_t$ are two connected components of $\tree-\{st\}$ that contain $s$ and $t$, respectively,
then $\sigma(st)$ separates $\bigcup_{x \in V_s} \beta(x) \setminus \sigma(st)$ from $\bigcup_{x \in V_t} \beta(x) \setminus \sigma(st)$ in $G$. 

The \emph{torso} of a bag $\beta(t)$ in a tree decomposition $(\tree,\beta)$
is a graph $H$ with $V(H) = \beta(t)$ and $uv \in E(H)$ if $uv \in E(G)$ or there exists a neighbor $s \in N_\tree(t)$ with $u,v \in \sigma(st)$. 
That is, the torso of $\beta(t)$ is created from $G[\beta(t)]$ by turning the adhesion $\sigma(st)$ into a clique for every neighbor $s$ of $t$ in $\tree$. 

\subparagraph{Extended strip decompositions.}
We follow the notation of~\cite{ICALP-qptas,arxiv-sttt-qpoly}.
For a graph $H$, by $T(H)$ we denote the set of triangles in $H$.
An \emph{extended strip decomposition} of a graph $G$ is a pair $(H, \eta)$ that consists of:
\begin{itemize}
\item a simple graph $H$,
\item a \emph{vertex set} $\eta(x) \subseteq V(G)$ for every $x \in V(H)$,
\item an \emph{edge set} $\eta(xy) \subseteq V(G)$ for every $xy \in E(H)$, and its subsets $\eta(xy,x),\eta(xy,y) \subseteq \eta(xy)$,
\item a \emph{triangle set} $\eta(xyz) \subseteq V(G)$  for every $xyz \in T(H)$,
\end{itemize}
which satisfy the following properties:
\begin{enumerate}
\item The family $\{\eta(o)~|~o \in V(H)\cup E(H) \cup T(H)\}$ is a partition of $V(G)$.
\item For every $x \in V(H)$ and every distinct $y,z \in N_H(x)$, the set $\eta(xy,x)$ is complete to $\eta(xz,x)$.
\item Every $uv \in E(G)$ is contained in one of the sets $\eta(o)$ for $o \in V(H) \cup E(H)\cup T(H)$, or is as follows:
\begin{itemize}
\item $u \in \eta(xy,x), v\in \eta(xz,x)$ for some $x \in V(H)$ and $y,z \in N_H(x)$, or
\item $u \in \eta(xy,x), v\in \eta(x)$ for some $xy \in E(H)$, or
\item $u \in \eta(xyz)$ and $v\in \eta(xy,x) \cap \eta(xy,y)$ for some $xyz \in T(H)$. 
\end{itemize}
\end{enumerate}
An extended strip decomposition $(H,\eta)$ is \emph{rigid} if for every $xy \in E(H)$, the sets $\eta(xy)$, $\eta(xy,x)$, and
$\eta(xy,y)$ are nonempty, and for every isolated $x \in V(H)$, the set $\eta(x)$ is nonempty.
Note that if $(H,\eta)$ is a rigid extended strip decomposition of $G$, then
$|V(H)|$ is bounded by $|V(G)|$.

For an extended strip decomposition $(H,\eta)$ of a graph $G$, we identify five \emph{types}
of \emph{particles}.
\begin{align*}
\textrm{vertex particle:} &\quad A_{x} := \eta(x) \text{ for each } x \in V(H),\\
\textrm{edge interior particle:} &\quad A_{xy}^{\perp} := \eta(xy) \setminus (\eta(xy,x) \cup \eta(xy,y)) \text{ for each } xy \in E(H),\\
\textrm{half-edge particle:} &\quad A_{xy}^{x} :=  \eta(x) \cup \eta(xy) \setminus \eta(xy,y) \text{ for each } xy \in E(H),\\
\textrm{full edge particle:} &\quad A_{xy}^{xy} := \eta(x) \cup \eta(y) \cup \eta(xy) \cup \bigcup_{z ~:~ xyz \in T(H)} \eta(xyz) \text{ for each } xy \in E(H),\\
\textrm{triangle particle:} &\quad A_{xyz} := \eta(xyz) \text{ for each } xyz \in T(H).
\end{align*}

As announced in the introduction, to solve \textsc{MWIS} in $G$ it suffices to know
the solution to \textsc{MWIS} in particles.
The proof of the following lemma follows closely the lines of proofs of analogous statement
of~\cite{DBLP:conf/soda/ChudnovskyPPT20} and is included for completeness in \Cref{app:matching}.

\begin{lemma}\label{lem:esd-border-mwis}
Given a \textsc{Border MWIS} instance $(G,\weight,T)$, an extended strip decomposition
$(H,\eta)$ of $G$, and a solution $f_{G[A],\weight,T \cap A}$ to the \textsc{Border MWIS} instance $(G[A],\weight,T \cap A)$
for every particle $A$ of $(H,\eta)$, 
one can in time $2^{|T|}$ times a polynomial in $|V(G)| + |V(H)|$ compute 
the solution $f_{G,\weight,T}$ to the input $(G,\weight,T)$. 
\end{lemma}

We need the following simple observations.
\begin{lemma}\label{lem:esd-deg}
Let $G$ be a $K_t$-free graph and let $(H,\eta)$ be a rigid extended strip decomposition of
$G$. Then the maximum degree of $H$ is at most $t-1$.
\end{lemma}
\begin{proof}
Let $x \in V(H)$. Observe that the sets
$\{\eta(xy,x)~|~y \in N_H(x)\}$ are nonempty and complete to each other in $G$.
Hence, $G$ contains a clique of size equal to the degree of $x$ in $H$. 
\end{proof}
\begin{lemma}\label{lem:esd-rep}
Let $G$ be a graph and let $(H,\eta)$ be an extended strip decomposition of
$G$ such that the maximum degree of $H$ is at most $d$.
Then, every vertex of $G$ is in at most $\max(4,2d+1)$ particles.
\end{lemma}
\begin{proof}
Pick $v \in V(G)$ and observe that:
\begin{itemize}
\item If $v \in \eta(x)$ for some $x \in V(H)$, then $v$ is in the vertex particle of $x$
and in one half-edge and one full-edge particle for every edge of $H$ incident with $x$. 
Since there are at most $d$ such edges, $v$ is in at most $2d+1$ particles.
\item If $v \in \eta(xy)$ for some $xy \in E(H)$, then $v$ is in at most four particles
for the edge $xy$. 
\item If $v \in \eta(xyz)$ for some $xyz \in T(H)$, then $v$ is in the triangle particle
for $xyz$ and in three full edge particles, for the three sides of the triangle $xyz$.
\end{itemize}
\end{proof}

%\section{General problem}\label{sec:problem}
%\input{problem}

\section{Bounded-degree graphs: Proof of Theorem~\ref{thm:main-maxdeg}}\label{sec:degree}
This section is devoted to the proof of~\cref{thm:main-maxdeg}.

Let $d,t$ be positive integers and let $(G,\weight)$ be the input vertex-weighted graph.
We denote $n := |V(G)|$ and $\maxdeg$ to be the maximum degree of $G$.
Let 
\[ \ell := (d-1)(3t+1) + \lceil 11 \log n + 6 \rceil (t + 2) = \Oh(dt + t \log n) \]
be an upper bound on
the size of $X$ for any application of~\cref{thm:ICALP:weight}
for any induced subgraph of $G$.

We describe a recursive algorithm that takes as input
an induced subgraph $G'$ of $G$ with weights $\weight$ and a set of terminals
$T \subseteq V(G')$ of size at most $4\ell \maxdeg^2$ and solves \textsc{Border MWIS}
on $(G',\weight,T)$. 
The root call is for $G' := G$ and $T := \emptyset$; indeed, note that
$f_{G,\weight,\emptyset}(\emptyset)$ is the maximum possible weight of an independent
set in $G$.

Let $(G',\weight,T)$ be an input to a recursive call. First, the algorithm initializes $f_{G',\weight,T}(I_T) := -\infty$ for every $I_T \subseteq T$. 

If $|V(G')| \leq 4 \maxdeg^2 \ell$, the algorithm proceeds by brute-force:
it enumerates independent sets $I \subseteq V(G')$ and updates
$f_{G',\weight,T}(I \cap T)$ with $\weight(I)$ whenever the previous value 
of that cell was smaller. 
As $\ell = \Oh(dt+ t \log n)$, this step takes $2^{\Oh(dt \maxdeg^2)}n^{\Oh(t \maxdeg^2)}$ time.
This completes the description of the leaf step of the recursion.

If $|V(G')| > 4 \maxdeg^2 \ell$, the algorithm proceeds as follows.
If $|T| \leq 3\maxdeg^2 \ell$, let $U := V(G')$, and otherwise, let $U := T$. 
The algorithm invokes~\cref{thm:ICALP:weight} on $G'$ and $U$.
If an induced $dS_{t,t,t}$ is returned, then it can be returned by the main algorithm as it is in particular an induced
subgraph of $G$. Hence, we can assume that we obtain a set $X \subseteq V(G)$ of size at most $\ell$
and an extended strip decomposition $(H,\eta)$ of $G^\ast := G'-N_{G'}[X]$ whose every particle contains
at most $|U|/2$ vertices of $U$.

Observe that as $|X| \leq \ell$ and the maximum degree of $G$ is $\maxdeg$, we have
$ |N_{G'}(N_{G'}[X])| \leq \maxdeg^2 \ell$.
Let $T^\ast := (T \cap V(G^\ast)) \cup N_{G'}(N_{G'}[X])$.
Note that we have 
$ T^\ast \subseteq V(G^\ast)$ and $|T^\ast| \leq 5\maxdeg^2 \ell$.
For every particle $A$ of $(H,\eta)$, invoke a recursive call 
on $(G^\ast_A := G^\ast[A],\weight,T^\ast_A := T^\ast \cap A)$, obtaining
$f_{G^\ast_A,\weight,T^\ast_A}$ (or an induced $S_{t,t,t}$, which can be directly returned). 
Use~\cref{lem:esd-border-mwis} to obtain a solution $f_{G^\ast,\weight,T^\ast}$
to \textsc{Border MWIS} instance $(G^\ast,\weight,T^\ast)$. 

Finally, iterate over every $I_T \subseteq T^\ast \cup N_{G'}[X]$ (note that
$T \subseteq T^\ast \cup N_{G'}[X]$) and, if $I_T$ is independent,
update the cell $f_{G',\weight,T}(I_T \cap T)$
with the value $\weight(I_T \setminus T^\ast) + f_{G^\ast,\weight,T^\ast}(I_T \cap T^\ast)$,
if this value is larger than the previous value of this cell.
This completes the description of the algorithm.

The correctness of the algorithm is immediate thanks to~\cref{lem:esd-border-mwis}
and the fact that $N_{G'}[X]$ is adjacent in $G'$ only to $N_{G'}(N_{G'}[X])$ which
is a subset of $T^\ast$. 

For the complexity analysis,
consider a recursive call to $(G^\ast_A,\weight,T^\ast_A)$ for a particle $A$. 
If $|T| \leq 3 \maxdeg^2 \ell$, then $|T^\ast_A| \leq |T^\ast| \leq 4 \maxdeg^2 \ell$.
Otherwise, $U = T$ and $|T \cap A| \leq |T|/2 \leq 2 \maxdeg^2 \ell$. 
As $|N_{G'}(N_{G'}[X])| \leq \maxdeg^2 \ell$, we have $|T^\ast_A| \leq 3 \maxdeg^2 \ell$.
Hence, in the recursive call the invariant of at most $4 \maxdeg^2 \ell$ terminals is maintained
and, moreover:
\begin{itemize}
\item if $|T| \leq 3 \maxdeg^2 \ell$, then $U = |V(G')|$ and $|V(G^\ast_A)| = |A| \leq |V(G')|/2$;
\item otherwise, $V(G^\ast_A) \subseteq V(G')$ and $|T^\ast_A| \leq 3 \maxdeg^2 \ell$, 
  hence the recursive call will fall under the first bullet.
\end{itemize}
We infer that the depth of the recursion is at most $2\lceil \log n \rceil$.

At every non-leaf recursive call, we spend $n^{\Oh(1)}$ time on invoking the algorithm from \cref{thm:ICALP:weight},
$2^{\Oh(dt \maxdeg^2)}n^{\Oh(t \maxdeg^2)}$ time
to compute $f_{G^\ast,\weight,T^\ast}$ using~\cref{lem:esd-border-mwis},
and $2^{\Oh(dt \maxdeg^2)}n^{\Oh(t \maxdeg^2)}$ time for the final iteration over all subsets $I_T \subseteq T^\ast \cup N_{G'}[X]$. 
Hence, the time spent at every recursive call is bounded by $2^{\Oh(dt \maxdeg^2)}n^{\Oh(t \maxdeg^2)}$.

At every non-leaf recursive call, we make subcalls to $(G^\ast_A,\weight,T_A^\ast)$ for
every particle $A$ of $(H,\eta)$. \Cref{lem:esd-deg,lem:esd-rep} ensure
that the sum of $|V(G^\ast_A)|$
over all particles $A$ is bounded by $(2\maxdeg+3) |V(G')|$. 
Hence, the total size of all graphs in the $i$-th level of the recursion is bounded
by $n \cdot (2\maxdeg+3)^i$. Since the depth of the recursion is bounded
by $2\lceil \log n \rceil$, the total size of all graphs in the recursion tree is bounded
by $n^{\Oh(\log \maxdeg)}$. Since this also bounds the size of the recursion tree,
we infer that the whole algorithm runs in time $2^{\Oh(dt \maxdeg^2)}n^{\Oh(t \maxdeg^2)}$. 

This completes the proof of~\cref{thm:main-maxdeg}.

\section{Graphs with no large bicliques: Proof of Theorem~\ref{thm:main-biclique}}\label{sec:degeneracy}
This section is devoted to the proof of \cref{thm:main-biclique}.

Let $t$ be a positive integer and let $k$ be the constant depending
on $t$ from~\cref{cor:tree-decomp}.
Again, let $(G,\weight)$ be the input vertex-weighted graph, let $n := |V(G)|$, and let 
\[ \ell := (d-1)(3t+1) + \lceil 11 \log n + 6 \rceil (t + 2) = \Oh(dt \log n)\]
be an upper bound on
the size of $X$ for any application of \cref{thm:ICALP:weight}
for any induced subgraph of $G$\footnote{As the dependence of $k$ on $d,t,s$ is superpolynomial, for the sake of simplicity, we do not try to optimize the dependence of the complexity bound on $d,t,s$.}.

The general framework and the leaves of the recursion are almost exactly the same
as in the previous section, but with different thresholds.
That is, 
we describe a recursive algorithm that takes as input
an induced subgraph $G'$ of $G$ with weights $\weight$ and a set of terminals
$T \subseteq V(G')$ of size at most $32k^5\ell$ and solves \textsc{Border MWIS}
on $(G',\weight,T)$. The root call is for $G' := G$ and $T := \emptyset$
and the algorithm returns $f_{G,\weight,\emptyset}(\emptyset)$ as the final answer.

Let $(G',\weight,T)$ be an input to a recursive call. The algorithm
initiates first $f_{G',\weight,T}(I_T) = -\infty$ for every $I_T \subseteq T$. 

If $|V(G')| \leq 32k^5\ell$, the algorithm proceeds by brute-force:
it enumerates independent sets $I \subseteq V(G')$ and updates
$f_{G',\weight,T}(I \cap T)$ with $\weight(I)$ whenever the previous value 
of that cell was smaller. 
As $\ell = \Oh(dt \log n)$ and $k$ is a constant depending on $d,t$, and $s$, this step takes polynomial time.
This completes the description of the leaf step of the recursion.

Otherwise, if $|V(G')| > 32k^5\ell$, we invoke~\cref{cor:tree-decomp}
on $G'$, obtaining a tree decomposition $(\tree,\beta)$ of $G'$. 
If $|T| \leq 24k^5 \ell$, let $U := V(G') \setminus T$, and otherwise, let $U := T$.

For every $t_1t_2 \in E(\tree)$, proceed as follows. For $i=1,2$, let $\tree_i$ be the connected
component of $\tree-\{t_1t_2\}$ that contains $t_i$ and let $V_i = \bigcup_{x \in \tree_i} \beta(x) \setminus \sigma(t_1t_2)$.
Clearly, $\sigma(t_1t_2)$ separates $V_1$ from $V_2$. Orient the edge $t_1t_2$ towards $t_i$ with larger $|U \cap V_i|$, breaking ties arbitrarily.

There exists $t \in V(\tree)$ of outdegree $0$. Then, 
for every connected component $C$ of $G'-\beta(t)$ we have $|C \cap U| \leq |U|/2$. 
%
%We check if there exists an edge $st \in E(\tree)$ such that the adhesion $\sigma(st)$
%is a balanced separator in the following sense: every connected component of $G'-\sigma(st)$
%contains at most $|U|/2$ vertices of $U$. 
%
%If this is the case, we perform a standard branch on $\sigma(st)$.
%That is, we iterate over all independent sets $J \subseteq \sigma(st)$
%and, if $\mathcal{C}^J$ is the set of connected components of $G'-\sigma(st)-N_{G'}(J)$
%for every $C \in \mathcal{C}^J$ we 
%recurse on $(G_C^J := G'[C],\weight,T_C^J := T \cap C)$,
%obtaining a function $f_{G_C^J,\weight,T_C^J}$. 
%Then, we iterate over all subsets $I_T$ of $\sigma(st) \cup T$,
%  set $J := I_T \cap \sigma(st)$,
%and update the cell $f_{G',\weight,T}(I_T \cap T)$ with 
%value $\weight(J) + \sum_{C \in \mathcal{C}^J} f_{G_C^J,\weight,T_C^J}(I_T \cap C)$
%if it is larger than the previous value of this cell. 
%
%If no such edge $st$ exists, there exists $t \in V(\tree)$ such that
%for every connected component $C$ of $G-\beta(t)$ we have $|C \cap U| < |U|/2$. 
Fix one such node $t$ and let $B := \beta(t)$ and let $\mathcal{C}$ be the set of
connected components of $G'-B$. 
Let $G^B$ be a supergraph of $G'[B]$ obtained from $G'[B]$ by turning, for every $C \in \mathcal{C}$,
the neighborhood $N_{G'}(C)$ into a clique. Note that $G^B$ is a subgraph
of the torso of $\beta(t)$. 
Hence, by the properties promised by~\cref{cor:tree-decomp},
  for every $C \in \mathcal{C}$ we have $|N_{G'}(C)| < k$ (as this set is contained
  in a single adhesion of an edge incident with $t$ in $\tree$)
  and
$G^B$ contains at most $k$ vertices of degree larger than $2k(k-1)$.
Let $Q$ be the set of vertices of $G^B$ of degree larger than $2k(k-1)$. 

We perform exhaustive branching on $Q$. That is, we iterate over all independent
sets $J \subseteq Q$ and denote $G^J := G'-Q-N_{G'}(J)$, $T^J := T \cap V(G^J)$,
     $U^J := U \cap V(G^J)$.
For one $J$, we proceed as follows.

We invoke~\cref{thm:ICALP:weight} to $G^J$ with set $U^J$,
obtaining a set $X^J$ of size at most $\ell$ 
and an extended strip decomposition $(H^J,\eta^J)$ of $G^J-N_{G^J}[X^J]$
whose every particle has at most $|U^J|/2 \leq |U|/2$ vertices of $U$.
Note that $G^J$ is an induced subgraph of $G'$, which is an induced subgraph of $G$,
     so there is no induced $dS_{t,t,t}$ in $G^J$.

A component $C \in \mathcal{C}$ is \emph{dirty} if $N_{G^J}[X^J] \cap N_{G'}[C] \neq \emptyset$ and \emph{clean} otherwise.
Let
\[ Y^J := (N_{G^J}[X^J] \cap B) \cup \bigcup_{C \in \mathcal{C} : C\mathrm{\ is\ dirty}} (N_{G'}(C) \cap V(G^J)). \]
The following bounds will be important for further steps.
\begin{equation}\label{eq:NXJ}
|N_{G^J}[X^J] \cap B| \leq (2k(k-1)+1)|X^J|.
\end{equation}
To see~\eqref{eq:NXJ} observe that a vertex $v \in X^J \cap B$ has at most $2k(k-1)$ neighbors in $B$ (as every vertex of $B \setminus Q$
    has degree at most $2k(k-1)$ in $G_B$) while every vertex $v \in X^J \setminus B$ has at most $k$ neighbors in $B$, as every component
of $G'-B$ has at most $k$ neighbors in $B$. This proves~\eqref{eq:NXJ}.
\begin{equation}\label{eq:YJ}
|Y^J| \leq (k + (2k(k-1)+1)^2) |X^J| \leq 4k^4 |X^J| \leq 4k^4 \ell = \Oh(k^4 dt \log n).
\end{equation}
To see~\eqref{eq:YJ}, consider a dirty component $C \in \mathcal{C}$. Observe that either $C$ contains a vertex of $X^J$ or 
$N_{G'}(C) \cap V(G^J)$ contains a vertex of $N_{G^J}[X^J] \cap B$. 
There are at most $|X^J|$ dirty components of the first type, contributing in total at most $k|X^J|$ vertices to $Y^J$. 
For the dirty components of the second type, although there can be many of them,
we observe that if $v \in N_{G'}(C) \cap N_{G^J}[X^J] \cap B$, then $N_{G'}(C) \cap V(G^J) \subseteq N_{G_B}[v]$.
Hence, for every dirty component of the second type, it holds
that $N_{G'}(C) \cap V(G^J) \subseteq N_{G_B}[N_{G^J}[X^J] \cap B]$. 
Since the degree of each vertex of $G_B$ is at most $2k(k-1)$, 
by~\eqref{eq:NXJ} we have 
\[ \left| N_{G_B}\left[N_{G^J}[X^J] \cap B \right] \right| \leq (2k(k-1)+1)^2 |X^J|. \]
The bound~\eqref{eq:YJ} follows.

A component $C \in \mathcal{C}$ is \emph{touched} if it is dirty or $N_{G'}(C)$ contains a vertex of $Y^J$. 
Let 
\[ Z^J := (N_{G^J}[Y^J] \cap B) \cup \bigcup_{C \in \mathcal{C} : C\mathrm{\ is\ touched}} N_{G'}(C) \cap V(G^J). \]
Using similar arguments as before, we can prove
\begin{equation}\label{eq:ZJ}
|Z^J| \leq (2k(k-1)+1)|Y^J| \leq 8k^5 |X^J| \leq 8k^5 \ell = \Oh(k^5 dt \log n).
\end{equation}
Indeed, if $C$ is touched, then $N_{G'}(C)$ contains a vertex $v \in Y^J$ (if $C$ is dirty, $N_{G'}(C) \cap V(G^J)$ is contained in $Y^J$), and then
$N_{G'}(C)$ is contained in $N_{G_B}[v]$. Also, for $v \in Y^J$ we have $N_{G^J}[v] \cap B \subseteq N_{G_B}[v]$.
Hence, $Z^J \subseteq N_{G_B}[Y^J]$. Since the maximum degree of a vertex of $B \setminus Q$ is $2k(k-1)$, this proves~\eqref{eq:ZJ}. 

For every touched $C \in \mathcal{C}$, denote $G_C := G^J[N_{G'}[C] \cap V(G^J)]$ and $T_C := ((T \cap C) \cup N_{G'}(C)) \cap V(G^J)$. 
Recurse on $(G_C,\weight,T_C)$, obtaining $f_{G_C,\weight,T_C}$.

Let 
\[ G^Y := G^J - Y^J - \bigcup_{C \in \mathcal{C} : C\mathrm{\ is\ touched}} C. \]
Note that, by the definition of dirty and touched, $G^Y$ is an induced subgraph of $G^J-N_{G^J}[X^J]$. 
Hence, $(H^J,\eta^J)$ can be restricted to a (not necessarily rigid) extended strip decomposition
$(H^J,\eta^{J,Y})$ of $G^Y$. 

Let $T^Y := (T \cup Z^J) \cap V(G^Y)$. 
For every particle $A$ of $(H^J,\eta^{J,Y})$,
recurse on $(G^Y[A],\weight,T^Y \cap A)$, obtaining $f_{G^Y[A],\weight,T^Y \cap A}$. 
Then, use these values with~\cref{lem:esd-border-mwis} to solve a \textsc{Border MWIS} instance $(G^Y,\weight,T^Y)$, obtaining $f_{G^Y,\weight,T^Y}$.

Iterate over every independent set $I_T \subseteq (T \cap V(G^J)) \cup T^Y \cup Y^J$.
Observe that $G'$ admits an independent set $I$ with $I \cap (Q \cup T \cup T^Y \cup Y^J) = J \cup I_T$ and weight:
\[ \weight(J) + \weight(I_T \setminus T^Y) + f_{G^Y,\weight,T^Y}(I_T \cap T^Y) + \sum_{C \in \mathcal{C} : C\mathrm{\ is\ touched}} \left(f_{G_C,\weight,T_C}(N_{G'}[C] \cap I_T) - \weight(I_T \cap N_{G'}(C))\right). \]
Update the cell $f_{G',\weight,T}((I_T \cup J) \cap T)$ with this value if it is larger than the previous value of this cell.
This finishes the description of the algorithm.

For correctness, it suffices to note that for every touched component $C$, the whole $N_{G'}(C) \cap V(G^J)$ is in the terminal set 
for the recursive call $(G_C,\weight,T_C)$ and the whole $N_{G'}(C) \cap V(G^Y)$ is in $Z^J$ and thus in the terminal set for the \textsc{Border MWIS} instance $(G^Y,\weight,T^Y)$. 

For the sake of analysis, consider a recursive call on $(G_C,\weight,T_C)$ for a touched component $C$.
If $|T| \leq 24k^5\ell$ and $U = V(G') \setminus T$, then $|T_C| \leq |T|+k \leq 32k^5\ell$ and $|V(G_C) \setminus T_C| \leq |C \setminus T| \leq |V(G') \setminus T|/2$. 
Otherwise, if $|T| > 24k^5\ell$ and $U=T$, then $|T_C| \leq |T|/2+k \leq 16k^5\ell + k \leq 24k^5\ell$. 
Thus, the recursive call on $(G_C,\weight,T_C)$ will fall under the first case of at most $24k^5\ell$ terminals.

Analogously, consider a recursive call on $(G^Y[A],\weight,T^Y \cap A)$ for a particle $A$ of $(H^J,\eta^{J,Y})$. 
If $|T| \leq 24k^5\ell$ and $U = V(G') \setminus T$, then $|T^Y \cap A| \leq |T^Y| \leq |T|+|Z^J| \leq 32k^5\ell$ 
due to~\eqref{eq:ZJ}. Furthermore, $|V(G^Y[A]) \setminus T^Y| \leq |V(G') \setminus T|/2$. 
Otherwise, if $|T| > 24k^5\ell$ and $U=T$, then $|T^Y \cap A| \leq |T|/2+|Z^J| \leq 16k^5\ell + 8k^5\ell \leq 24k^5\ell$ again due to~\eqref{eq:ZJ}.
Thus, the recursive call on $(G^Y[A],\weight,T^Y \cap A)$ will fall under the first case of at most $24k^5\ell$ terminals.

Finally, note that a recursive call $(G',\weight,T)$ without nonterminal vertices (i.e., with $T = V(G')$) is a leaf call.

We infer that all recursive calls satisfy the invariant of at most $32k^5\ell$ terminals and the depth of the recursion tree is bounded by $2 \lceil \log n \rceil$ (as every second
level the number of nonterminal vertices halves).

At each recursive call, we iterate over at most $2^k$ subsets $J \subseteq Q$.
\cref{lem:esd-deg} ensures that the maximum degree of $H^J$ is at most $2t-1$, while~\cref{lem:esd-rep} ensures that every vertex of $G^Y$ is used in at most
$4t$ particles of $(H^J,\eta^{J,Y})$. 
In a subcall $(G_C,\weight,T_C)$ for a touched component $C$, vertices of $C$ are not used in any other call for the current choice of $J$, while all vertices of $V(G_C) \setminus C$ are terminals.
Consequently, every nonterminal vertex $v$ of $G'$ is passed as a nonterminal vertex to a recursive subcall at most $2^k \cdot 4t$ number of times
(and a terminal is always passed to a subcall as a terminal). 
Furthermore, a recursive call without nonterminal vertices is a leaf call. 
As the depth of the recursion is $\Oh(\log n)$, we infer that, summing over all recursive calls in the entire algorithm, the number of nonterminal vertices
is bounded by $n^{\Oh(\log t + k)}$ and the total size of the recursion tree is $n^{\Oh(\log t + k)}$. 

At each recursive call, we iterate over all $2^k$ subsets $J \subseteq Q$ and then we invoke~\cref{thm:ICALP:weight}
and iterate over all independent sets $I_T$ in $(T \cap V(G^J)) \cup T^Y \cup Y^J$.
Thanks to the invariant $|T| \leq 32k^5\ell$ and bounds~\eqref{eq:YJ}, and~\eqref{eq:ZJ}, this set is of size $\Oh(k^5 \ell)$. 
Hence, every recursive call runs in time $n^{\Oh(k^5 t) + k c_{d,t}}$, where $c_{d,t}$ is a constant depending on $d$ and $t$. 
As $k$ is a constant depending on $d,t,s$, the final running time bound is polynomial.

This completes the proof of~\cref{thm:main-biclique}.

\section{Conclusion}\label{sec:conclusion}
%Let us remark that the authors of~\cite{ICALP-qptas} conjecture that the size of the set $X$ in~\cref{thm:ICALP:weight} can be actually made constant. The main combinatorial insight of~\citep{ACDR21} is that such a statement holds for bounded-degree graphs.

%Let us remark that it is very easy to lift~\cref{thm:main-maxdeg} to the class of bounded-degree $H$-free graphs, where every component of $H$ is a subdivided claw, see~\cite{ACDR21,ICALP-qptas}. However, it is far from clear  how to do it in the setting of \cref{thm:main-biclique}.

While it is generally believed that \MWIS is polynomial-time-solvable in $S_{t,t,t}$-free (and even $dS_{t,t,t}$-free) graphs (with no further assumptions), such a result seems currently out of reach. Thus it is interesting to investigate how further can we relax the assumptions on instances, as we did when going from \cref{thm:main-maxdeg} to \cref{thm:main-biclique}.
In particular, we used the assumption of $K_r$-freeness twice: once in \cref{lem:lozinrazgon} and then to argue that $H$ (the pattern of an extended strip decomposition we obtain) is of bounded degree. On the other hand, the assumption of $K_{r,r}$-freeness was used just once: in \cref{lem:lozinrazgon}.
Thus it seems natural to try to prove the following conjecture.

\begin{conjecture}
For every integers $t,r$ there exists a polynomial-time algorithm that, given 
an $S_{t,t,t}$-free and $K_r$-free vertex-weighted graph $(G,\weight)$
computes the maximum possible weight of an independent set in $(G,\weight)$.
\end{conjecture}

\paragraph{Acknowledgements} We acknowledge the welcoming and productive atmosphere at Dagstuhl Seminar 22481 “Vertex Partitioning in Graphs: From Structure to Algorithms,” where some part of the work leading to the results in this paper was done.

\bibliography{main}

\appendix

\section{Appendix: Proof of Lemma~\ref{lem:lozinrazgon}}\label{app:lozinrazgon}
For positive integers $a,b$, the Ramsey number of $a$ and $b$, denoted by $\mathsf{Ram}(a,b)$, is the smallest integer $r$ such that every graph on $r$ vertices contains either an independent set of size $a$ or a clique of size $b$.
It is well-known that $\mathsf{Ram}(a,b) \leq \binom{a+b-2}{a-1}$.

For a graph $H$, a \emph{subdivision of $H$} is any graph obtained from $H$ by subdividing each edge arbitrarly many times (possibly 0).
By a \emph{$t$-subdivision} (resp., \emph{$(\leq t)$-subdivision}) we mean a subdivision where each edge was subdivided exactly (resp., at most $t$) times.
A \emph{proper subdivision} is one where each edge was subdivided at least once.
By $S(p,t)$ we denote the $t$-subdivision of the $p$-leaf star. In particular, $S(3,t)$ is exactly $S_{t,t,t}$.

We will need the following two technical results shown by Lozin and Razgon~\cite{DBLP:journals/ejc/LozinR22}.

\begin{lemma}[{{\cite[Lemma 2]{DBLP:journals/ejc/LozinR22}}}]\label{lem:lozinrazgon1}
There is a function $c : \N \times \N \to \N$ with the following property.
If a graph $G$ contains a collection of $c(a,s)$ pairwise disjoint subsets of vertices, each of size at most $a$ and with at least one edge between any two of them, then $G$ contains $K_{s,s}$ as a subgraph.
\end{lemma} 

\begin{lemma}[{{\cite[Theorem 3]{DBLP:journals/ejc/LozinR22}}}]\label{lem:lozinrazgon2}
There is a function $m : \N \times \N \to \N$ with the following property.
Every graph $G$ containing a $(\leq t)$-subdivision of $K_{m(h,t)}$ as a subgraph contains either $K_{s,s}$ as a subgraph or a proper $(\leq t)$-subdivision of $K_{h,h}$ as an induced subgraph.
\end{lemma}

The following lemma is the main technical ingredient of the proof. It can be seen as a strengthening of Claim 2 in~\cite{DBLP:journals/ejc/LozinR22}.

\begin{lemma}\label{lem:lozinrazgonMain}
For all positive integers $d,t,r,s$ there exists $k=k(d,t,r,s)$ such that the following holds.
If $G$ contains a $k$-block, then it contains:
\begin{enumerate}
\item a $(\leq t)$-subdivision of $K_r$ as a subgraph or
\item $K_{s,s,}$ as a subgraph, or
\item  $dS_{t,t,t}$ as an induced subgraph.
\end{enumerate}
\end{lemma}
\begin{proof}
Let $c$ be the functions given by \Cref{lem:lozinrazgon1}.
Define the following constants
\begin{align*}
 q = &~ 2 \; \mathsf{Ram}(d,c(3t+1,s)), \\
 p = &~ (q/2-1)(3t+1)+3,\\
 \ell = &~ \mathsf{Ram}(p,c(t,s)), \\
 k' = &~ \max \left (r+q, t\binom{r}{2} + \ell \right),\\
 k = &~  2k'.
\end{align*}

Suppose that $G$ has a $k$-block $B$ but no subgraph isomorphic to $K_{s,s}$ nor induced subgraph isomorphic to $dS_{t,t,t}$.
We aim to show that $G$ has a $(\leq t)$-subdivision of $K_r$ as a subgraph.

Consider a pair $\{x,y\}$ of distinct vertices from $B$.
Since $x$ and $y$ cannot be separated by deleting fewer than $k$ vertices,
by Menger's theorem there is a set $\mathcal{P}(x,y)$ of $k$ internally pairwise disjoint $x$-$y$-paths in $G$.
Without loss of generality we can assume that each such path is induced.

\begin{claim}\label{clm:isolateTerminals}
There is a subset $B' \subseteq B$ of size $k'=k/2$ with the following property. 
For each pair of distinct vertices $x,y \in B'$ there is a set $\mathcal{P'}(x,y)$ of $k'$ internally pairwise disjoint induced $x$-$y$-paths in $G$ that do not contain any vertices from $B' \setminus \{x,y\}$.
\end{claim}
\begin{claimproof}
Let $B'$ be any set of $k'$ vertices from $B$.
Consider any two distinct vertices $x,y \in B'$.
As the interiors of paths in $\mathcal{P}(x,y)$ are pairwise disjoint,
we observe that at most $k'$ paths might contain vertices from $B' \setminus \{x,y\}$.
Thus $\mathcal{P}(x,y)$ contains at least $k - k' = k'$ paths that do not intersect $B' \setminus \{x,y\}$.
\end{claimproof}

Fix a pair $\{x,y\}$ of distinct vertices from $B'$.
A path in $\mathcal{P'}(x,y)$ is \emph{long} if it has at least $t$ internal vertices.
The pair $\{x,y\}$ is \emph{distant} if at least $\ell$ paths in $\mathcal{P'}(x,y)$ are long.

Let $Q \subseteq B'$ be a minimum-size subset of $B'$ that intersects all distant pairs.
(It might be useful to think of $Q$ as the minimum vertex cover in a graph with vertex set $B'$ where the edges correspond to distant pairs.)
We consider two cases, depending whether $Q$ is ``large'' or ``small''.

\paragraph{\boldmath Case 1 (large $Q$): $|Q| \geq q$.}
In this case there is a set consisting of at least $|Q|/2 \geq  \mathsf{Ram}(d,c(3t+1,s))$  pairwise disjoint distant pairs (i.e., a matching in the mentioned graph).
Let $M$ be such a set of size exactly $\mathsf{Ram}(d,c(3t+1,s))$.

\begin{claim}\label{clm:extractWideClaw}
For each $\{x,y\} \in M$, there is a set $S^{\{x,y\}}$ contained in the union of paths in $\mathcal{P'}(x,y)$ that induces $S(p,t)$ in $G$.
\end{claim}
\begin{claimproof}
Fix $\{x,y\} \in M$. 
For a long path $P \in \mathcal{P'}(x,y)$, let $\mathsf{prefix}(P)$ be the set of $t$ first vertices of $P$, starting from the side of $x$, but excluding $x$ itself. So $\mathsf{prefix}(P)$ induces a $t$-vertex path, and $\mathsf{prefix}(P) \cup \{x\}$ induces a $(t+1)$-vertex path in $G$.

Fix any set $\mathcal{P'}_{\text{long}}$ of $\ell$ long paths from $\mathcal{P'}(x,y)$, and define $Z = \{ \mathsf{prefix}(P) ~|~ P \in \mathcal{P'}_{\text{long}} \}$.
Consider an auxiliary graph $\mathbf{Z}$ with vertex set $Z$, in which two elements are adjacent if and only if there is an edge between one set and the other (we emphasize here that the sets in $Z$ are pairwise disjoint).

If $\mathbf{Z}$ contains a clique of size $c(t,s)$, then, by~\cref{lem:lozinrazgon1}, we obtain a subgraph isomorphic to $K_{s,s}$ in $G$, a contradiction.
Thus, since $|Z| = \ell =  \mathsf{Ram}(p,c(t,s))$ we obtain an independent set of size $p$ in $\mathbf{Z}$. 
The $p$ paths forming this independent set, together with $x$, induce the desired copy of $S(p,t)$.
\end{claimproof}

Observe that even though the pairs in $M$ are pairwise disjoint, the sets $S^{\{x,y\}}$ might still intersect.
In the next claim we extract from them induced copies of $S_{t,t,t}$ that are pairwise disjoint.

\begin{claim}\label{clm:extractDisjointClaws}
For each $\{x,y\} \in M$ there is an induced copy of $S_{t,t,t}$ contained in $S^{\{x,y\}}$,
such that for any distinct $\{x,y\}, \{x',y'\} \in M'$ the corresponding copies of $S_{t,t,t}$ are disjoint.
\end{claim}
\begin{claimproof}
Fix an arbitrary order $\{x_1,y_1\},\ldots,\{x_{|M|},y_{|M|}\}$ on pairs in $M$. The proof is by induction.

The copy of $S_{t,t,t}$ for $\{x_1,y_1\}$ can be selected by picking any three out of $p$ paths in the copy of $S(p,t)$ given by~\cref{clm:extractWideClaw}.
So let $i \in [1,|M|-1]$ and suppose that for each $1 \leq j \leq i$ we have selected an induced copy $S^j$ of $S_{t,t,t}$ contained in $S^{\{x_j,y_j\}}$.
Note that in total we have selected $i \cdot (3t+1)$ vertices.
Consider the pair $\{x_{i+1},y_{i+1}\}$ and let $P^1,\ldots,P^p$ be the $t$-vertex paths obtained from $S^{\{x_{i+1},y_{i+1}\}}$ by deleting $x_{i+1}$.
Note that at most $i(3t+1) \leq (|M|-1)(3t+1)$ of these paths might intersect $S^1 \cup \ldots \cup S^i$. 
So, as $p = (|M|-1)(3t+1)+3$, there is always a choice of three paths that are disjoint with $S^1 \cup \ldots \cup S^i$.
Recall that by \cref{clm:isolateTerminals} the vertex $x_{i+1}$ is not contained $\bigcup_{1 \leq j \leq i} S^j$.
Thus the vertices of the three selected paths, together with $x_{i+1}$, form the set $S^{i+1}$.
\end{claimproof}

Now we proceed similarly as in the proof of \cref{clm:extractWideClaw}.
Let $S$ be the family consisting of induced copies of $S_{t,t,t}$ for all $\{x,y\} \in M$; they are given by \cref{clm:extractDisjointClaws}.
Let $\mathbf{S}$ be the graph with vertex set $S$ where two vertices are adjacent if and only if there is an edge from one set to another.
Note that $|S| = \mathsf{Ram}(d,c(3t+1,s))$. 
An independent set in $\mathbf{S}$ of size $d$ corresponds to an induced $dS_{t,t,t}$ in $G$.
On the other hand, if $\mathbf{S}$ has a clique of size at least $c(3t+1,s)$, then by \cref{lem:lozinrazgon1} we obtain $K_{s,s}$ as a subgraph of $G$. By the choice of $|S|$ one of these cases must happen, and both yield a contradiction. Thus we conclude that Case 1 cannot occur.

\paragraph{\boldmath Case 2 (small $Q$): $Q < q$.}
Define $B'' = B' \setminus Q$; note that $|B''| \geq k' - q \geq r$.
For any distinct $x,y \in B''$, let $\mathcal{P}''(x,y)$ be obtained from $\mathcal{P}'(x,y)$ by removing all long paths.
By the definition of $Q$ we observe that for all $x,y$ we have $|\mathcal{P}''(x,y)| \geq k' - \ell \geq t \binom{r}{2}$.

Let $R$ be any $r$-element subset of $B''$.

\begin{claim}\label{clm:pathsinclique}
For each pair $\{x,y\}$ of distinct vertices of $R$ there is a path $P^{\{x,y\}} \in \mathcal{P''}(x,y)$ such the paths selected for distinct pairs are internally disjoint.
\end{claim}
\begin{claimproof}
The proof is similar to the proof of~\cref{clm:extractDisjointClaws}. We use induction.
Enumerate pairs of distinct vertices from $R$ as $\{x_1,y_1\},\ldots,\{x_{\binom{r}{2}},y_{\binom{r}{2}}\}$.

The path $P^{\{x_1,y_1\}}$ can be arbitrarily chosen from $\mathcal{P''}(x_1,y_1)$.
Suppose that we have selected paths $P^{\{x_1,y_1\}},\ldots,P^{\{x_i,y_i\}}$ for some $1 \leq i < \binom{r}{2}$.
Since each path is short, the selected paths have in total at most $i \cdot t < t\binom{r}{2}$ internal vertices.

Now consider the set $\mathcal{P''}(x_{i+1},y_{i+1})$. Recall that the paths in this set are pairwise disjoint. Since $|\mathcal{P''}(x_{i+1},y_{i+1})| \geq t \binom{r}{2}$, we observe that there is a path in $\mathcal{P''}(x_{i+1},y_{i+1})$ which is internally disjoint with all previously selected paths. We pick this path as $P^{\{x_{i+1},y_{i+1}\}}$.
\end{claimproof}

Recall that for each distinct $x,y \in R$, the path $P^{\{x,y\}}$ does not contain vertices from $R \setminus \{x,y\}$.
Thus the set $R$ with the paths given by \cref{clm:pathsinclique} forms a $(\leq t)$-subdivision of $K_r$ which is a subgraph of $G$.
This concludes the claim.
\end{proof}

Now we can proceed to the proof of~\cref{lem:lozinrazgon}.

\begin{proof}[of~\cref{lem:lozinrazgon}.]
Define $h = d(3t+1)$ (i.e., the number of vertices of $dS_{t,t,t}$).
Define $r = m(h,t)$, where $m$ is as in~\cref{lem:lozinrazgon2}, and let $k = k(d,t,r,s)$ be as in \cref{lem:lozinrazgonMain}.
For contradiction, suppose that $G$ has a $k$-block.
By~\cref{lem:lozinrazgonMain} we conclude that $G$ contains a $(\leq t)$-subdivision of $K_r$ as a subgraph.
By~\cref{lem:lozinrazgon2} we observe that $G$ contains a proper $(\leq t)$-subdivision of $K_{h,h}$ as an induced subgraph.

It is straightforward to verify that this induced subgraph contains a subdivision of $dS_{t,t,t}$ (in fact, of any graph with at most $h$ vertices and at most $h$ edges) as an induced subgraph. Therefore $G$ contains an induced subgraph isomorphic to $dS_{t,t,t}$, a contradiction.
\end{proof}

\section{Appendix: Proof of Lemma~\ref{lem:esd-border-mwis}}\label{app:matching}
Iterate over every $I_T \subseteq T$.
For fixed $I_T$, we aim at computing $f_{G,\weight,T}(I_T)$.
If $I_T$ is not independent, we set $f_{G,\weight,T}(I_T) = -\infty$.
In the remainder of the proof, we show how to compute in polynomial time
the value $f_{G,\weight,T}(I_T)$ for fixed independent $I_T \subseteq T$.

For a particle $A$ of $(H,\eta)$, let $a(A) := f_{G[A],\weight,T \cap A}(I_T \cap A)$
and let $I(A)$ be an independent set witnessing this value, that is, an independent set 
in $G[A]$ of weight $\weight(a(A))$ with $I(A) \cap T \cap A = I_T \cap A$. 
Note that as $I_T$ is independent, the value $a(A)$ is not equal to $-\infty$ and such an independent
set exists.

We say that $x \in V(H)$ is \emph{forced} if $I_T \cap \bigcup_{y \in N_H(x)} \eta(xy,x) \neq \emptyset$.
Note that since $I_T$ is independent, if $x$ is forced, then $\eta(xy,x) \cap I_T \neq \emptyset$ for exactly one edge $xy$ incident with $x$.
We call such an edge $xy$ the \emph{enforcer} of $x$. Note that an edge $xy$ may be the enforcer of both $x$ and $y$.

The arguments now follow very closely the outline of Section~3.3 of~\cite{DBLP:journals/corr/abs-1907-04585}.

We construct a set $\mathcal{P}$ of particles and an edge-weighted graph $(H',\weight')$ as follows.
We start with $\mathcal{P} = \emptyset$, $V(H') = V(H)$, and $E(H') = \emptyset$.

For every $x \in V(H)$ that is not forced, add $A_x$ to $\mathcal{P}$.
For every $xyz \in T(H)$ such that neither of the edges $xy$, $yz$, or $xz$ is the enforcer of \emph{both} its endpoints, add $A_{xyz}$ to $\mathcal{P}$.
For every $e = xy \in E(H)$, proceed as follows. 
\begin{enumerate}
\item If neither $x$ nor $y$ is forced, we add to $H'$ a new vertex $t_e$
and edges $t_ex$, $t_ey$, $xy$, and set the edge weights $\weight'$ as follows:
\begin{align*}
\weight'(t_ex) &:= a(A_{xy}^x) - a(A_{xy}^\bot) - a(A_x),\\
\weight'(t_ey) &:= a(A_{xy}^y) - a(A_{xy}^\bot) - a(A_y),\\
\weight'(xy) &:= a(A_{xy}^{xy}) - a(A_{xy}^\bot) - a(A_x) - a(A_y) - \sum_{z, \text{ s.t. }xyz \in T(H)} a(A_{xyz}).
\end{align*}
Furthermore, add $A_{xy}^\bot$ to $\mathcal{P}$. 
\item If exactly one of $x$ and $y$ is forced, say w.l.o.g. $x$ is forced and $y$ is not forced, proceed as follows.
\begin{enumerate}
\item If $xy$ is the enforcer of $x$, then add to $H'$ an edge $xy$ with weight
\[ \weight'(xy) := a(A_{xy}^{xy}) - a(A_{xy}^x) - a(A_y) - \sum_{z, \text{ s.t. }xyz \in T(H)} a(A_{xyz}).\]
Furthermore, add $A_{xy}^x$ to $\mathcal{P}$. 
\item If $xy$ is not the enforcer of $x$, then add to $H'$ a new vertex $t_e$ and an edge $t_ey$ with weight
\[ \weight'(t_ey) := a(A_{xy}^y) - a(A_{xy}^\bot) - a(A_y). \]
Furthermore, add $A_{xy}^\bot$ to $\mathcal{P}$.
\end{enumerate}
\item If both $x$ and $y$ are forced, proceed as follows.
\begin{enumerate}
\item If $xy$ is neither the enforcer of $x$ nor of $y$, add $A_{xy}^\bot$ to $\mathcal{P}$.
\item If $xy$ is the enforcer of $x$ but not of $y$ add $A_{xy}^x$ to $\mathcal{P}$.
\item If $xy$ is the enforcer of $y$ but not of $x$ add $A_{xy}^y$ to $\mathcal{P}$.
\item If $xy$ is the enforcer of both $x$ and $y$, add $A_{xy}^{xy}$ to $\mathcal{P}$.
\end{enumerate}
\end{enumerate}

This finishes the description of the construction of $\mathcal{P}$ and $(H',\weight')$.
In the next two paragraphs we make two observations that follow by a direct check from the definitions.

Observe that $I_0 := \bigcup_{A \in \mathcal{P}} I(A)$ is independent in $G$ 
and has weight $a_0 := \sum_{A \in \mathcal{P}} a(A)$. 
Furthermore, for every $A \in \mathcal{P}$, we have $I_0 \cap A \cap T = I_T \cap A$
and $I_0 \cap T = I_T$.
We think of $I_0$ as the ``base'' solution for $f_{G,\weight,T}(I_T)$.

Observe also that all weights $\weight'$ of $H'$ are nonnegative, as $A_{xy}^x$ contains both $A_{xy}^\bot$ and $A_x$
while $A_{xy}^{xy}$ contains $A_{xy}^\bot$, $A_x$, $A_y$, as well as all $A_{xyz}$ for 
all triangles $xyz$ containing the edge $xy$. 

We will be asking for a maximum weight matching in $(H',\weight')$. 
Intuitively, taking an edge $t_ex$ to such a matching corresponds to replacing 
in $I_0$ the parts $I(A_{xy}^\bot)$ and $I(A_x)$ with the part $I(A_{xy}^x)$
while taking an edge $xy$ to such a matching corresponds to replacing
in $I_0$ the parts $I(A_{xy}^\bot)$, $I(A_x)$, $I(A_y)$ and all parts $I(A_{xyz})$ for triangles
$xyz$ containing the edge $xy$ with part $I(A_{xy}^{xy})$. 

From another perspective, fix $x \in V(H)$ and recall that the sets $\eta(xy,x)$ for $y \in N_H(x)$
are complete to each other. Hence, any independent set in $G$ can contain vertices in at most
one of such sets. For an edge $e = xy \in E(H)$, taking an edge $xy$ or $t_ex$ in a matching
in $H'$ corresponds to choosing that, among all neighbors of $x$ in $H$, the neighbor $y$
is such that the set $\eta(xy,x)$ is allowed to contain vertices of the sought independent set.
(Choosing $xy \in E(H')$ to the matching corresponds to allowing both $\eta(xy,x)$ and $\eta(yx,y)$
 to contain vertices of the sought independent set.)

However, there is a delicacy if $I_T$ contains a vertex of some interface $\eta(xy,x)$.
Then, in some sense $I_T$ already forces some choices in the corresponding matching in $H'$.
This is modeled above by having alternate construction for vertices $x \in V(H)$ that are forced.

The following two claims prove that $f_{G,\weight,T}(I_T)$ equals $a_0$ plus
the maximum possible weight of a matching in $(H',\weight')$ and thus
complete the proof of~\cref{lem:esd-border-mwis}.
Their proofs follow exactly the lines of the proofs of Claims~3.7 and~3.8 of Section~3.3
of~\cite{DBLP:journals/corr/abs-1907-04585} and are thus omitted.

\begin{claim}
Let $I$ be an independent set in $G$ with $I \cap T = I_T$. 
Let $M$ be the set of edges of $H'$ defined as follows: 
for every $e = xy \in E(H)$,
 if $\eta(xy,x) \cap I \neq \emptyset$ and $\eta(xy,y) \cap I \neq \emptyset$, then $xy \in M$,
 if $\eta(xy,x) \cap I \neq \emptyset$ and $\eta(xy,y) \cap I = \emptyset$, then $t_ex \in M$,
and 
 if $\eta(xy,x) \cap I = \emptyset$ and $\eta(xy,y) \cap I \neq \emptyset$, then $t_ey \in M$.
Then, all the above edges indeed exist in $H'$ and  $M$ is a matching.
Furthermore, the weight of $I$ is at most $a_0 + \sum_{e \in M} \weight'(e)$.
\end{claim}

\begin{claim}
Let $M$ be a matching in $H'$. 
Let $\mathcal{P}_M$ be the set of particles of $(H,\eta)$ defined as follows.
Start with $\mathcal{P}_M := \mathcal{P}$ and then for every edge $e = xy \in E(H)$, 
  \begin{itemize}
  \item if $xy \in M$, insert $A_{xy}^{xy}$ into $\mathcal{P}_M$ 
  and remove from $\mathcal{P}_M$ the following particles if present: $A_{xy}^x$, $A_{xy}^y$, $A_{xy}^\bot$, $A_x$, $A_y$, $A_{xyz}$ for any $z \in V(H)$ such that $xyz \in T(H)$.
  \item if $t_ex \in M$ (resp. $t_ey \in M$), insert $A_{xy}^x$ (resp. $A_{xy}^y$) into $\mathcal{P}_M$,
  and remove from $\mathcal{P}_M$ the following particles if present: $A_{xy}^\bot$ and $A_x$ (resp. $A_y$).
  \end{itemize}
Then $I_M := \bigcup_{A \in \mathcal{P}_M}$ is an independent set in $G$ with $I_M \cap T = I_T$
and of weight at least $a_0 + \sum_{e \in M} \weight'(e)$. 
\end{claim}

\end{document}